\documentclass[conference]{IEEEtran}%
\usepackage{amsmath}
\usepackage{bbm}
\usepackage{graphicx}
\usepackage{epsfig}
\usepackage{array}
\usepackage{psfrag}
\usepackage{amssymb}
\usepackage{mathdots}
\usepackage{color}
\usepackage{pstricks,pst-node,pst-text,pst-3d,pst-plot}
\usepackage{psfrag}
\usepackage{enumerate}
\usepackage{url,cite}
\usepackage{amsfonts}%
\usepackage{amsthm}
\usepackage{dsfont}%
\usepackage{verbatim}%
\usepackage{setspace}
\usepackage{float}
\usepackage{url}
\usepackage{fancyhdr}
\usepackage{bm}
\usepackage{lipsum} 
\newcommand\blfootnote[1]{%
  \begingroup
  \renewcommand\thefootnote{}\footnote{#1}%
  \addtocounter{footnote}{-1}%
\endgroup
}
\usepackage{algorithm}
\usepackage{algorithmic}

\newtheorem{cond}{Condition}
\newtheorem{Def}{Definition}

\newcommand{\bRopt}{R^{\rm (opt)}_i}

\newcommand{\Oneik}{\mathds{1}_i(k)}
\newcommand{\Ai}{A_i}

\newcommand{\Rik}{R_i(k)}

\newcommand{\Qik}{Q_i(k)}

\newcommand{\Yik}{Y_i(k)}
\newcommand{\bfYk}{{\mathbf Y}(k)}
\newcommand{\bfQk}{{\mathbf Q}(k)}
\newcommand{\LW}[1]{W_0\lb #1e^{-1}\rb}
\newcommand{\gammaikm}{\gamma_i^{(m)}(k)}
\newcommand{\gammaik}{\gamma_i(k)}

\newcommand{\NR}{N_{\rm R}}
\newcommand{\NNR}{N_{\rm NR}}

\newcommand{\sM}{\script{M}}
\newcommand{\sN}{\script{N}}
\newcommand{\Nk}{N_k}
\newcommand{\sNk}{\script{N}_k}
\newcommand{\sNR}{\script{N}_{\rm R}}
\newcommand{\sNRk}{\script{N}_{\rm R}\lb k\rb}

\newcommand{\sNNR}{\script{N}_{\rm NR}}

\newcommand{\SRk}{\script{S}_{\rm R}\lb k\rb}
\newcommand{\SRstk}{\script{S}_{\rm R}^*\lb k\rb}

\newcommand{\SRkst}{\script{S}_{\rm R}^*\lb k\rb}
\newcommand{\SNRk}{\script{S}_{\rm NR}\lb k\rb}
\newcommand{\bfPRk}{\bfP\lb k\rb}

\newcommand{\PRik}{P_i\lb k\rb}

\newcommand{\PRistk}{P_{\ist}\lb k\rb}
\newcommand{\PNRik}{P_i\lb k\rb}
\newcommand{\bmuRi}{\overline{R}_i}

\newcommand{\bfmuRk}{{\bm \mu}\lb k\rb}

\newcommand{\muRik}{\mu_i\lb k\rb}

\newcommand{\muRistk}{\mu_{\ist}\lb k\rb}

\newcommand{\parPRdef}[1]{\triangleq \left[#1_i(k)\right]_{i\in\sN}}

\newcommand{\PsiRik}{\Psi_{\rm R}(i,k)}
\newcommand{\PsiNRik}{\Psi_{\rm NR}(i,k)}
\newcommand{\PsiNRikst}{\Psi_{\rm NR}^*(i,k)}
\newcommand{\PsiNRistkst}{\Psi_{\rm NR}^*(\ist,k)}

\newcommand{\ist}{i_{\rm NR}^*}

\newcommand{\bfrk}{{\bm r}\lb k\rb}

\newtheorem{thm}{Theorem}
\newtheorem{lma}{Lemma}
\DeclareMathOperator{\E}{\mathbb{E}}
\newenvironment{proofsketch}{\par{\it Proof Sketch:}}{\qed\par}

\newcommand{\bP}{\overline{P}}

\newcommand{\lb}{\left (}
\newcommand{\rb}{\right )}

\newcommand{\script}[1]{{\mathcal {#1}}}
\newcommand{\Pavg}{P_{\rm avg}}
\newcommand{\Pmax}{P_{\rm max}}

\newcommand{\EE}[1]{\E \left[ #1 \right]}
\newcommand{\EEU}[1]{\E_{\bfU(k)} \left[ #1 \right]}
\newcommand{\EEY}[1]{\E_{\bfU(k)} \left[ #1 \right]}

\newcommand{\bfP}{{\bf P}}

\newcommand{\bgamma}{\overline{\gamma}}

\newcommand{\bfQ}{{\bf Q}}

\newcommand{\bfY}{{\bf Y}}
\newcommand{\bfU}{{\bf U}}

\newcommand{\parRdef}[1]{\triangleq [#1_1(k),\cdots,#1_{\NR}(k)]^T}

\newcommand{\Rmax}{R_{\rm max}}

\newcommand{\Ts}{T}

\begin{document}
\title{Power Control and Scheduling under Hard Deadline Constraints for On-Off Fading Channels}

\author{Ahmed Ewaisha, Cihan Tepedelenlio\u{g}lu\\
\small{School of Electrical, Computer, and Energy Engineering, Arizona State University, USA.}\\
\small{Email:\{ewaisha, cihan\}@asu.edu}\\
}
\maketitle
\blfootnote{The work in this paper has been supported by NSF Grant CCF-1117041.}
\begin{abstract}
We consider the joint scheduling-and-power-allocation problem of a downlink cellular system. The system consists of two groups of users: real-time (RT) and non-real-time (NRT) users. Given some average power constraint on the base station, the problem is to find an algorithm that satisfies the RT and NRT quality-of-service (QoS) constraints. The RT QoS constraints guarantee the portion of RT packets that miss their deadline are no more than a pre-specified threshold. On the other hand, the NRT QoS is only to guarantee the stability of their queues. We propose a sum-rate-maximizing algorithm that satisfy all QoS and average power constraints. The proposed power allocation policy has a closed-form expression for the two groups of users. However, the power policy of the RT users differ in structure from the NRT users. The proposed algorithm is optimal for the on-off channel model with a polynomial-time scheduling complexity. Using extensive simulations, the throughput of the proposed algorithm is shown exceed existing approaches.
\end{abstract}

\section{Introduction}

Quality-of-service-based scheduling has gained strong attention recently. It is shown in \cite{Lai20131689} and \cite{piro2011two} that quality-of-service-aware scheduling results in a better performance in LTE systems compared to quality-of-service-unaware techniques. Depending on the application, quality-of-service (QoS) metrics may refer to long-term throughput \cite{6848162}, short-term throughput \cite{hsiehheavy}, per-user average delay \cite{ewaisha2015joint}, average number of packets missing a specific deadline \cite{A_Theory_of_QoS}, or the average time a user waits to receive its data \cite{hou2015qoe}. Real-time applications, such as audio and video applications, need to be served by algorithms that takes average packet delays or the probability of a packet missing the deadline into consideration. This is because these applications have stringent requirements for the service times of their packets. If a packet is not scheduled to be transmitted on time, the corresponding user might experience intermittent connectivity of its audio or video.

The problem of scheduling for wireless systems under a hard deadline constraint has been widely studied in the literature (see, e.g., \cite{hou2011survey} and \cite{radhakrishnan2016review} for a survey). In \cite{A_Theory_of_QoS} the authors consider binary erasure channels and present a sufficient and necessary condition to determine if a given problem is feasible. The work is extended in three different directions. The first direction studies the problem under delayed feedback \cite{piro2011two}. The second considers general channel fading models. An example is \cite{hou2010scheduling} that present a scheduling algorithm that guarantees a pre-specified portion of the packets to be transmitted by the deadline. The third direction studies multicast video packets that have strict deadlines and utilize network coding to improve the overall network performance \cite{Hou:2015:BDT:2823437.2823441,Adaptive_NC_Deadline}. Unlike the time-framed assumption in the previous works, the authors of \cite{kang2013performance} assume that arrivals and deadlines do not have to occur at the edges of a time frame and present a scheduling algorithm with a fixed power allocation. In \cite{Elastic_Inelastic} the authors study the scheduling problem in the presence of real-time and non-real-time data. Unlike real-time data, non-real-time data do not have a strict deadline but rather can be transmitted at any arbitrary point in time. However, there is an implicit constraint that the queues of the non-real-time data need to be stable. Using the dual function approach, the problem was decomposed into an online algorithm that guarantees network stability and real-time users' satisfaction.

Power allocation has not been considered for RT users in the literature, to the best of our knowledge. In this paper, we study the problem of resource allocation in the presence of simultaneous RT and NRT users in a downlink cellular system. We formulate the problem as a joint scheduling-and-power-allocation problem to maximize the sum throughput of the NRT users subject to an average power constraint on the base station (BS), as well as a delivery ratio requirement constraint for each RT user. The delivery ratio constraint requires a minimum ratio of packets to be transmitted by a hard deadline, for each RT user. Perhaps the closest to our work are references \cite{Elastic_Inelastic} and \cite{Ewaisha_TVT2015}. The former does not consider power allocation, while the latter assumes that only one user can be scheduled per time slot. 
Our contributions in this paper are as follows:
\begin{itemize}
	\item We present a rate-maximizing joint scheduling-and-power-allocation algorithm. We show that this algorithm satisfies the average power constraint and delivery ratio requirement constraint.
	\item We present closed-form expressions for the power allocation policy used by our algorithm. It is shown that the power allocation expressions for the RT and NRT users are different in structure.
	\item We show that the complexity of our optimal scheduling algorithm is polynomial in the number of users in the network.
\end{itemize}

The rest of this paper is organized as follows. In Section \ref{Model} we present the system model and the underlying assumptions. The problem is formulated in Section \ref{Problem_Formulation}. For the on-off channel model, the proposed power-allocation and scheduling algorithm as well as its optimality is presented in Section \ref{Proposed_Algorithm}. Simulation results and comparisons with baseline approaches is presented in Section \ref{Results}. Finally, the paper is concluded in Section \ref{Conclusion}.


\section{System Model}
\label{Model}
We assume a time slotted downlink system with a single base station (BS) and a single frequency channel. There are $N$ users in the system indexed by the set $\sN\triangleq\{1, \cdots,N\}$. The set of users is divided into two sets: the RT set of users $\sNR\triangleq\{1\cdots\NR\}$ and the NRT set of users $\sNNR\triangleq\{\NR+1\cdots \NR+\NNR\}$ with $\NR$ and $\NNR$ denoting the number of RT and NRT users, respectively. Following \cite{A_Theory_of_QoS}, we model the channel between the BS and the $i$th user as a fading channel with gain $\gammaik=1$ if it is in a ``good'' state during the $k$th slot and $\gammaik=0$ otherwise following a Bernoulli process. Channel gains are fixed over the whole slot and change independently in subsequent slots and are independent across users. Moreover, the channel state information for all users are known to the BS at the beginning of the each slot.

\subsection{Packet Arrival Model}
Following \cite{hou2010scheduling} we assume that the duration of each slot $\Ts$ is long enough that more than one user can be scheduled in this slot. Let $a_i(k)\in\{0,1\}$ be the indicator of a packet arrival for user $i\in\sN$ at the beginning of the $k$th slot. $\{a_i(k)\}$ is assumed to be a Bernoulli process with rate $\lambda_i$ packets per slot and assumed to be independent across all users in the system. We assume that packets arriving for the RT users are real-time packets while those for the NRT users are non-real-time ones. Real-time packets have a strict transmission deadline. If a packet is not transmitted by this deadline, this packet is dropped out of the system and does not contribute towards the throughput of the user. Here we assume that real-time packets arriving at the beginning of the $k$th slot have their deadline at the end of this slot. On the other hand, non-real-time packets do not have a strict deadline. Thus each packet remains in the (infinite-sized \cite{Bertsekas_Data_Networks}) buffer until it is completely transmitted.

\subsection{Service Model}
\begin{figure}%
\centering
\includegraphics[width=1\columnwidth]{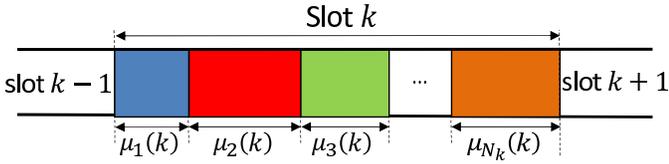}%
\caption{In the $k$th time slot, the BS chooses $\Nk$ users to be scheduled. All time slots have a fixed duration of $\Ts$ seconds.}%
\label{Time_Slot}%
\end{figure}
At the beginning of the $k$th slot, the BS selects a set of RT users denoted $\SRk\subseteq\sNR$ and a set of NRT users denoted $\SNRk\subseteq\sNNR$. Thus a total of $\Nk\triangleq\vert\sNk\vert$ users are scheduled at slot $k$ where $\sNk\triangleq\SRk\cup\SNRk$ (Fig. \ref{Time_Slot}). Moreover, the BS assigns an amount of power $\PRik$ for every user $i\in\sNk$. This dictates the transmission rate for each scheduled users according to the capacity of the channel given by $\Rik=\log \lb 1+\PRik\gammaik\rb$, $\forall i\in\sN$. Finally, the BS determines the duration of time, out of the $\Ts$ seconds, that will be allocated for each scheduled user. Define the variable $\muRik$ to represent the duration of time, in seconds, assigned for user $i\in\sN$ during the $k$th slot (Fig. \ref{Time_Slot}). Hence, $\muRik\in[0,\Ts]$ for all $i\in\sN$. The BS decides the value of this variable for each user $i\in\sN$ at the beginning of slot $k$. Unlike NRT users which do not have to transmit their packets at a particular time slot, RT users have a strict deadline. Hence, if an RT user was scheduled at slot $k$, then it should be allocated the channel for a duration of time that allows the transmission of the whole packet. Thus we have
\begin{equation}
\muRik=\left\{
\begin{array}{lll}
	\frac{L}{\Rik} &\mbox{ if }i\in\SRk\\
	0 &\mbox{ if }i\in\sNR\backslash\SRk
\end{array}
\right.,
\label{Num_Slots}
\end{equation}
where $L$ is the number of bits per packet, that is assumed to be fixed for all packets in the system. Equation \eqref{Num_Slots} means that, depending on the transmission power, if RT user $i$ is scheduled at slot $k$, then it is assigned as much time as required to transmit its whole packet. Hence, given \eqref{Num_Slots}, $\muRik$ is not an optimization variable and it becomes sufficient for the BS to decide the value of $\PRik$ and the set $\SRk$ for the RT users. On the other hand, for the NRT users, the BS needs to optimize over $\{\muRik\}_{i\in\sNNR}$ as well. Thus the queue associated with user $i\in\sNNR$ is given by
\begin{equation}
Q_i(k+1)=\lb \Qik + Lr_i(k)-\muRik\Rik\rb^+
\label{Queues}
\end{equation}
where $r_i(k)$ is the admission control decision variable at the beginning of slot $k$. The BS sets $r_i(k)=1$ when two conditions are satisfied: 1) if $a_i(k)=1$ and 2) if the BS decides to admit this arriving packet to the $i$th buffer. Otherwise the BS sets $r_i(k)=0$. The time-average number of packets admitted for user $i$ is
\begin{equation}
\Ai\triangleq \limsup_{K\rightarrow \infty}\frac{1}{K}\sum_{k=1}^K \EE{r_i(k)}, \hspace{0.25in} i\in\sNNR
\label{Avg_Admit}\\
\end{equation}

The BS's goal is solve this power allocation and scheduling problem along with the admission control decisions to maximize the NRT users' sum rate under the system constraints. In the next section we present this problem formally.

\section{Problem Formulation}
\label{Problem_Formulation}
In this work, we are interested in finding the scheduling and power allocation algorithm that maximizes the sum-rate of all NRT users subject to the system constraints. In this paper we restrict our search to slot-based algorithms which, by definition, takes the decisions only at the beginning of the slots. Since the channel coefficients do not change within a slot, restricting our search to the space of slot-based algorithms does not lose optimality.

Now define the time-average rate, in packets per slot, of user $i$ to be
\begin{equation}
\bmuRi\triangleq \liminf_{K\rightarrow \infty}\frac{1}{LK}\sum_{k=1}^K\muRik \Rik, \hspace{0.25in} i\in\sNNR
\label{Avg_Rate}\\
\end{equation}
while the time-average power consumed by the BS is $\bP\triangleq\limsup_{K\rightarrow \infty}\frac{1}{K}\sum_{k=1}^K P(k)$ where $P(k)$ is the power consumed by the BS during the $k$th slot which is given by $P(k)\triangleq \frac{1}{\Ts}\sum_{i\in\sN}\PRik \muRik$. Thus the problem we are interested to solve in this paper is to find the scheduling, power allocation and packets admission decisions, at the beginning of each slot, that solve the following problem
\begin{align}
\label{Prob_DL}
\text{maximize} &\sum_{i\in\sNNR}\bmuRi,\\
\text{subject to } & r_i(k)\leq a_i(k) \hspace{0.25in} \forall i\in\sNNR,
\label{Admission_Decision}\\
& \bmuRi\geq\Ai \hspace{0.25in}i\in\sNNR,
\label{NRT_QoS}\\
&\bmuRi\geq\lambda_i q_i \hspace{0.15in}i\in\sNR,
\label{RT_QoS}\\
&\bP\leq \Pavg,
\label{P_avg}\\
& 0\leq\PRik\leq \Pmax \hspace{0.25in} \forall i\in\sN,
\label{P_max}\\
&\sum_{i\in\sN}\muRik= \Ts, \hspace{0.25in} k\geq 1,
\label{Single_Tx_at_a_Time}\\
&0\leq\muRik\leq\Ts \hspace{0.25in} \forall i\in\sN,
\label{muRik_Constr_NRT}\\
\nonumber\text{variables } &\{\bfmuRk,\bfPRk,\bfrk\}_{k=1}^\infty,
\end{align}
where $\bfmuRk\parPRdef{\mu}$ while $\bfPRk\parPRdef{P}$. Constraint \eqref{Admission_Decision} says that no packets should be admitted to the $i$th buffer if no packets arrived to user $i$. Constraint \eqref{NRT_QoS} indicates that the average service rate for the NRT users has to be higher than the average number of packets admitted to the buffer. Constraint \eqref{RT_QoS} indicates that the resources allocated to a RT user $i$ need to be such that the fraction of packets transmitted by the deadline are greater than the required QoS $q_i$. Constraint \eqref{P_avg} is an average power constraint on the BS transmission power. Finally constraint \eqref{Single_Tx_at_a_Time} guarantees that the sum of durations of transmission of all scheduled users doesn't exceed the slot duration $\Ts$. In this paper, we assume that the NRT user with the longest queue has enough packets, at each slot, to fit the whole slot duration which is a valid assumption in the heavy traffic regime. It will be clear that generalizations to the non-heavy traffic regime is possible by allowing multiple NRT users to be scheduled but this is omitted for brevity.
\vspace{-0.1in}
\section{Proposed algorithm}
\label{Proposed_Algorithm}
We solve this problem using the Lyapunov optimization technique \cite{li2011delay} where each average constraint in problem \eqref{Prob_DL} is associated with a ``virtual queue'' which represents how much this average constraint is not satisfied. We discuss these virtual queues next.
\subsection{Virtual Queues}
We define a virtual queue associated with each RT user as follows
\begin{equation}
Y_i(k+1)=\lb \Yik + a_i(k)q_i-\Oneik\rb^+, \hspace{0.25in} i\in\sNR,
\label{DL_VQ}
\end{equation}
where $\Oneik\triangleq\mathds{1}\lb\muRik\rb$ with $\mathds{1}(\cdot)=1$ if its argument is non-zero and $\mathds{1}(\cdot)=0$ otherwise. $\Yik$ is a measure of how much user $i$ is not satisfying constraint \eqref{RT_QoS}. For notational convenience we denote $\bfY(k)\parRdef{Y}$. We will later show a sufficient condition on $Y_i(k)$ for constraint \eqref{RT_QoS} to be satisfied. Hence, we say that the queue $Y_i(k)$ is associated with constraint \eqref{RT_QoS}. Similarly, we define the virtual queue $X(k)$, associated with constraint \eqref{P_avg}, as
\begin{equation}
X(k+1)=\lb X(k) + P(k)-\Pavg\rb^+
\label{P_avg_VQ}
\end{equation}
We observe that the virtual queues $\bfY(k)$ and $X(k)$ are analogous to the real queues $\bfQ(k)$. The latter can be thought of as a queue that is associated with the constraint \eqref{NRT_QoS}. To provide a sufficient condition on the queues (virtual or real) for the corresponding constraints to be satisfied, we use the definition of \emph{mean rate stability} of queues as in \cite[Definition 1]{li2011delay} to state the following lemma.

\begin{lma}
\label{Mean_Rate_Lemma}
If, for some $i\in\sNNR$, $\{Q_i(k)\}_{k=0}^\infty$ is mean rate stable, then constraint \eqref{NRT_QoS} is satisfied for this user $i$.
\end{lma}
\begin{proofsketch}
Proof follows along the lines of \cite[Lemma 1]{ewaisha2015joint} and is omitted here due to space limitations.
\end{proofsketch}
Lemma \ref{Mean_Rate_Lemma} shows that when the virtual queue $\Qik$ is mean rate stable, then constraint \eqref{NRT_QoS} is satisfied for user $i\in\sNNR$. Similarly, if $\{Y_i(k)\}_{k=0}^\infty$ and $\{X(k)\}_{k=0}^\infty$ are mean rate stable, then constraints \eqref{RT_QoS} and \eqref{P_avg} are, respectively, satisfied. Thus, our objective would be to devise an algorithm that guarantees the mean rate stability of $\Qik$ for all RT users, $\Yik$ for all NRT users as well as $X(k)$.
\subsection{Motivation of the Proposed Algorithm}
\label{Motivation_DL}
Following the Lyapunov optimization technique as in \cite{li2011delay}, we define the Lyapunov function
\begin{equation}
L\lb U(k)\rb\triangleq \frac{1}{2}\sum_{i\in\sNR}{Y_i^2(k)}+\frac{1}{2}\sum_{i\in\sNNR}{Q_i^2(k)}+\frac{1}{2}X^2(k),
\label{Lyapunov_Func}
\end{equation}
where $U(k)\triangleq \lb\bfYk,\bfQk,X(k)\rb$, and the Lyapunov drift as $\Delta (k) \triangleq \EEY{L(k+1) - L(k)}$, where $\EEU{x}\triangleq \EE{x\vert U(k)}$ is the conditional expectation of the random variable $x$ given $U(k)$. Squaring \eqref{Queues}, \eqref{DL_VQ} and \eqref{P_avg_VQ} taking the conditional expectation then summing over $i$, the drift becomes bounded by
\begin{equation}
\Delta(k)\leq C_1+\Psi(k),
\label{Drift_Bound}
\end{equation}
where $C_1\triangleq C/2$ with $C\triangleq \sum_{i\in\sNR}\lb q_i^2+1\rb+\Pmax^2+\Pavg^2+\NNR\left[ L^2+\Ts^2\Rmax^2\right]$ and we use $\Rmax\triangleq\log\lb1+\Pmax\rb$, while
\begin{multline}
\Psi(k)\triangleq \EEU{\sum_{i\in\sNR}\PsiRik}+\sum_{i\in\sNR}\Yik\lambda_i q_i\\
-X(k)\Pavg+\EEU{\sum_{i\in\sNNR}\PsiNRik\muRik}\\
+\sum_{i\in\sNNR}L\Qik r_i(k).
\label{Drift_Min}
\end{multline}
where $\PsiRik$ and $\PsiNRik$ are given by
\begin{equation}
\PsiRik\triangleq \lb\Yik-\frac{LX(k)\PRik}{\Ts\Rik}\rb\Oneik, \hspace{0.05in} i\in\sNR,\\
\label{PsiRik}
\end{equation}
\begin{equation}
\PsiNRik\triangleq \Qik\Rik-\frac{X(k)\PRik}{\Ts}, \hspace{0.05in} i\in\sNNR,
\label{PsiNRik}
\end{equation}
respectively, where we used \eqref{Num_Slots} in \eqref{PsiRik}. The proposed algorithm minimizes the first two terms of \eqref{Drift_Min} by scheduling the users and allocating their powers at each slot to solve
\begin{equation}
\begin{array}{ll}
	&\text{max}\sum_{i\in\SRk}\PsiRik + \sum_{i\in\sNNR}\PsiNRik\muRik\\
	&\text{subject to } \eqref{P_max}, \eqref{Single_Tx_at_a_Time} \text{ and } \eqref{muRik_Constr_NRT}
\end{array}
\label{Max_Prob}
\end{equation}

Problem \eqref{Max_Prob} is a joint power allocation and scheduling problem. To solve this mixed-integer programming problem optimally, we first find the optimal power-allocation-and-scheduling policy for the NRT users through the following lemma. Then we proceed with the RT users.

\begin{lma}
\label{NRT_Lemma_On_Off}
If an NRT user $i$ is scheduled to transmit any of its NRT data during the $k$th slot, then the optimum power level for this NRT with respect to (w.r.t.) problem \eqref{Max_Prob} is given by
\begin{equation}
\PNRik=\min\lb\lb\frac{Q_i(k)}{X(k)}-1\rb^+,\Pmax\rb, \hspace{0.25in} i\in\sNNR.
\label{H2O_Pow_On_Off}
\end{equation}
Moreover, in the heavy traffic regime, if an NRT user is going to be scheduled at slot $k$, then the optimum user w.r.t. problem \eqref{Prob_DL} is given by
\begin{equation}
\ist=\arg\max_{i\in\sNNR}\Qik
\label{ist}
\end{equation}
with ties broken arbitrarily.
\end{lma}
\begin{proof}
Proof is omitted due to lack of space.
\end{proof}

Lemma \ref{NRT_Lemma_On_Off} provides the optimal scheduling policy for the NRT users, at the $k$th slot, as well as the optimal power allocation w.r.t. problem \eqref{Max_Prob}. The lemma shows that if any of the NRT users is going to be scheduled in the $k$th slot, then only one of them is going to be scheduled. This means that the scheduling policy for the NRT users is
\begin{equation}
\muRik=\left\{
\begin{array}{lll}
	&\Ts-\sum_{i\in\SRstk}\muRik & i=\ist\\
	& 0 & \sNNR\backslash\{\ist\}
\end{array}
\right.
\label{mu_ist}
\end{equation}
which is a manipulation of \eqref{Single_Tx_at_a_Time}. Substituting \eqref{mu_ist} in \eqref{Max_Prob}, the latter becomes
\begin{equation}
\begin{array}{lll}
	&\max&\sum_{i\in\SRk}\PsiRik+\PsiNRistkst\muRistk\\
	&\text{subject to } &\eqref{P_max}, \eqref{mu_ist} \text{ and } \eqref{muRik_Constr_NRT}.
\end{array}
\label{Max_Prob_RT}
\end{equation}
where
\begin{equation}
\PsiNRikst\triangleq \Qik\log \lb\frac{\Qik}{X(k)}\rb-\Qik+X(k).
\label{PsiNRikst_On_Off}
\end{equation}

To find the scheduler of the RT users that is optimal w.r.t. problem \eqref{Max_Prob_RT}, we present the following lemma that has a lower complexity compared to the exhaustive search.
\begin{lma}
\label{NRik_Lemma_On_Off}
Given the optimal set $\SRkst$ of RT users that solves problem \eqref{Max_Prob}, if $i\in\SRkst$, $j\notin\SRkst$, $\gammaik=\gamma_j(k)=1$ and $\PRik=P_j(k)$ then $\Yik\geq Y_j(k)$.
\end{lma}
\begin{proof}
We prove this lemma by contradiction. Suppose $i\in\sNRk$ and $j\notin\sNRk$ and suppose that $\Yik<Y_j(k)$. We can increase the objective function of \eqref{Max_Prob} by swapping the two users. This swapping increases the objective function since $\Yik<Y_j(k)$ and the quantity $X(k)\PRik \muRik=X(k)P_j(k) \mu_j(k)$ based on the fact that $\PRik=P_j(k)$.
\end{proof}

Lemma \ref{NRik_Lemma_On_Off} says that if the power allocation policy results in a equal power allocation for all RT users, then there will be no scheduled RT users having a value of $Y_j(k)$ smaller than any of the unscheduled RT users. This lemma suggests an algorithm to reduce the complexity of scheduling the RT users from $O\lb2^{\NR}\rb$ to $O\lb \NR\rb$. This algorithm is achieved by listing the RT users in a descending order of their $Y_i(k)$. Without loss of generality, in the remaining of this paper we will assume that $Y_1>Y_2\cdots >Y_{\NR}$. We present the following definition then present a theorem that discusses a necessary condition for the optimum power allocation policy for the RT users.

\begin{Def}
\label{Lambert_Pow_Policy}
We define the Lambert power allocation policy for the RT users as
\begin{equation}
\PRik=\frac{\frac{\Ts\Psi_{\rm NR}^*(\ist,k)}{X(k)}-1}{\LW{\left[\frac{\Psi_{\rm NR}^*(\ist,k)\Ts}{X(k)}-1\right]}}-1, \hspace{0.1in} i\in\SRk,
\label{Lambert_Pow_On_Off}
\end{equation}
where $W_0(z)$ is the principle branch of the Lambert W function \cite{Lambert_W_Function} while $\PsiNRikst$ is given in \eqref{PsiNRikst_On_Off}.
\end{Def}

\begin{thm}
\label{Pow_Alloc_Nec_Thm}
Given any scheduled set of RT users $\SRk$, if the Lambert power policy results in $\sum_{i\in\SRk}L/\log(1+\PRik)\leq\Ts$, then it is the optimum RT-users' power allocation policy. Otherwise, the optimum power allocation policy is given by
\begin{equation}
\PRik=e^{\frac{\vert\SRk\vert L}{\Ts}}-1,\hspace{0.25in} i\in\SRk,
\label{RT_NO_NRT_Pow}
\end{equation}
and no NRT users should be scheduled in slot $k$.
\end{thm}
\begin{proofsketch}
This theorem is proved by applying the Lagrange optimization \cite[Ch. 5]{cvx_Boyd} technique to problem \eqref{Max_Prob_RT} then using the complementary slackness condition.

\end{proofsketch}
Under the Lambert power policy let's define $l$ as the number of RT users can be scheduled in slot $k$under the Lambert policy. $l$ is given by
\begin{equation}
l\triangleq\left\lfloor\frac{\Ts}{\muRik}\right\rfloor
\label{Critical_Index}
\end{equation}
Before presenting the algorithm that solves problem \eqref{Max_Prob_RT} and the theorem behind it we present the following two conditions on $l$ that will facilitate the understanding of the algorithm and the presentation of the theorem. 
\begin{cond}
$l=0$.
\label{Large_Duration_Cond1}
\end{cond}
\begin{cond}
$0<l\leq\NR$ and the following two inequalities hold $Y_l(k)>\left[X(k)P_l(k)+\PsiNRistkst\right]\mu_l(k)$ and $Y_{l+1}(k)\geq X(k) \left[\Ts \lb e^{\frac{(l+1)L}{\Ts}}-1\rb-P_l(k)l\mu_l(k)\right]
+\PsiNRistkst\lb \Ts-l\mu_l(k)\rb$
\label{Large_Duration_Cond2}
\end{cond}
Condition \ref{Large_Duration_Cond1} means that the duration $\muRik$ of one RT user is greater than the slot duration under the Lambert policy. On the other hand, Condition \ref{Large_Duration_Cond2} means that, roughly speaking, the $\Yik$ values are very high to the extent that the RT users are suffering more than the NRT users during slot $k$. The next theorem shows that when any of Conditions \ref{Large_Duration_Cond1} or \ref{Large_Duration_Cond2} holds, the BS should schedule only RT users during slot $k$.

\begin{thm}
\label{Thm_of_Algorithm_On_Off}
To solve \eqref{Max_Prob_RT}, if either Condition \ref{Large_Duration_Cond1} or Condition \ref{Large_Duration_Cond2} holds, then the optimal scheduling for the RT users is given by
\begin{equation}
\SRk=\left\{ i:Y_i>X(k)\Ts \lb e^{\frac{iL}{\Ts}}-e^{\frac{(i-1)L}{\Ts}}\rb\right\}.
\label{RT_NO_NRT_Sch}
\end{equation}
On the other hand if neither of these conditions holds, then the optimal scheduling policy for the RT users is
\begin{equation}
\SRk=\left\{ i: Y_i>\left[X(k)\PRik+\PsiNRistkst\right]\muRik\right\}
\label{Lambert_Sch_On_Off}
\end{equation} while that of the NRT users is given by \eqref{ist} and \eqref{mu_ist}.
\end{thm}

\begin{proof}
The proof is omitted due to lack of space.
\end{proof}


\subsection{Proposed Algorithm}
We now propose Algorithm \ref{Scheduling_Alg} which is the scheduling and power allocation algorithm for problem \eqref{Prob_DL}. Then we present the motivation and optimality of this algorithm in Sections \ref{Motivation_DL} and \ref{Optimality_DL_Section}, respectively. Algorithm \ref{Scheduling_Alg} is executed at the beginning of the $k$th slot and, without loss of generality, it assumes: 1) all RT users in the system have received a packet at the beginning of the $k$th slot, 2) all NRT users have non-empty buffers, and 3) all users in the system have an ``on'' channel. If, at some slot, any of these assumptions does not hold for some users, these users are eliminated from the system for this slot.

\begin{algorithm}
\caption{Scheduling and Power Allocation Algorithm}
\begin{algorithmic}[1]
\label{Scheduling_Alg}
\STATE Sort the RT users in a descending order of $\Yik$. Without loss of generality, assume that $Y_1>Y_2\cdots>Y_{\NR}$.
\STATE Find the user $\ist$ according to \eqref{ist}.
\STATE Set the power according to \eqref{Lambert_Pow_On_Off} for all RT users.
\STATE Calculate $\muRik$ and $l$ using \eqref{Num_Slots} and \eqref{Critical_Index}, respectively.
\IF {Condition \ref{Large_Duration_Cond1} OR Condition \ref{Large_Duration_Cond2} holds}
\STATE Set $\muRik=0$ for all $i\in\sNNR$ and set the scheduling and power allocation of the RT users according to \eqref{RT_NO_NRT_Sch} and \eqref{RT_NO_NRT_Pow}, respectively.
\ELSE
\STATE Schedule the RT users according to
\begin{equation}
\muRik=\left\{
\begin{array}{lll}
	&1 & i\in\SRk\\
	&0 & \mbox{otherwise}
\end{array}
\right.
\label{mu_RT}
\end{equation}
where $\SRk$ is given in \eqref{Lambert_Sch_On_Off} and set the RT users' powers according to \eqref{Lambert_Pow_On_Off}.
\STATE Schedule the NRT users according to \eqref{mu_ist} and set user $\ist$'s power $\PRistk$ via \eqref{H2O_Pow_On_Off}.
\ENDIF
\STATE Set $r_i(k)=a_i(k)$ if $Q_i(k)<V$ and $0$ else, $i\in\sNNR$.
\STATE Update \eqref{Queues}, \eqref{DL_VQ} and \eqref{P_avg_VQ} at the end of the $k$th slot.
\end{algorithmic}
\end{algorithm}

\subsection{Optimality of Proposed Algorithm}
\label{Optimality_DL_Section}
We first define $\bRopt$ to be the throughput of NRT user $i$ under the optimal algorithm that solves \eqref{Prob_DL}. The following theorem gives a bound on the performance of Algorithm \ref{Scheduling_Alg} compared to the optimal algorithm that has a genie-aided knowledge of $\bRopt$ which, we show that, due to this knowledge it can solve the problem optimally. 

\begin{thm}
\label{Optimality_Thm}
If $\gammaikm\in\{0,1\}$ for all $i\in\sN$, $k\geq1$ and all $m\in\sM$, then for any $V>0$ and any $\epsilon\in(0,1]$ Algorithm \ref{Scheduling_Alg} results in satisfying all constraints in \eqref{Prob_DL} and achieves an average rate satisfying
\begin{equation}
\sum_{i\in\sNNR} \bmuRi\geq \sum_{i\in\sNNR}{\bRopt} - \frac{C_1}{LV}.
\label{Optimality_Eq}
\end{equation}
\end{thm}

\begin{proofsketch}
We divide the proof into two parts. First, we show that the queues (real and virtual) are mean rate stable. This proves that constraints \eqref{NRT_QoS}, \eqref{RT_QoS} and \eqref{P_avg} are satisfied. Second, through the Lyapunov optimization technique we show that the drift-minus-reward term is within a constant gap from the performance of the optimal, genie-aided algorithm \cite{georgiadis2006resource,urgaonkar2011optimal}. The details of the proof is omitted due to lack of space.
\end{proofsketch}

Theorem \ref{Optimality_Thm} says that Algorithm \ref{Scheduling_Alg} yields an objective function \eqref{Prob_DL} that is arbitrary close to the performance of the optimal genie-aided algorithm that solves \eqref{Prob_DL}.

\section{Simulation Results}
\label{Results}
We assume that all channels are statistically homogeneous, i.e. $\bgamma_i=\bgamma$ for all $i\in\sN$ where $\bgamma$ is a fixed constant. Moreover, all RT users have homogeneous delivery ratio requirements, thus $q_i=q$ for all $i\in\sNR$ for some parameter $q$. All parameter values are summarized in Table \ref{Parameters} for all simulation figures unless otherwise specified.

We compare the throughput of the RT users, which is the objective of problem \eqref{Prob_DL}, to that of a simple power allocation and scheduling algorithm that we call ``FixedP'' algorithm. In the FixedP algorithm, all scheduled users transmit with the maximum power, i.e. $\PRik=P_{\rm max}$ for all $i\in\sN$ and all $k\geq1$, while the scheduling policy is to flip a biased coin and choose to schedule either the NRT users or the RT users. The coin is set to schedule the RT users with probability $q$ (the delivery ratio requirement for all users), at which case the RT users are sorted according to $Y_i(k)$ and scheduled one by one until the current slot ends. On the other hand, when the coin chooses the NRT users, the FixedP policy assigns the entire time slot to the NRT user with the longest queue.

\begin{table}
	\centering
		\caption{Simulation Parameter Values}
		\label{Parameters}
		\begin{tabular}{|c|c||c|c|}
			\cline{1-4}
			Parameter & Value & Parameter & Value \\
			\cline{1-4}
			$L$ & $1$ bit/packet &  $\Pmax$ & 200\\
			$V$ & $10^4$ & $\bgamma_i$, $\forall i$ & $1$\\
			$\left\{q_i\right\}_{i\in\sNR}$& $0.3$ & $\Ts$ & $1$\\
			\cline{1-4}
			\end{tabular}
\end{table}

We assume that we have $N=20$ users that is split equally between the RT and NRT users, i.e. $\NR=\NNR=20$. Fig. \ref{Algorithm1_vs_FixedP_P2_10} shows a substantial increase in the average rate of the proposed algorithm over the FixedP algorithm with over $200\%$ at low $\Pavg$ values and $60\%$ at high $\Pavg$ values.

\begin{figure}%
\centering
\includegraphics[width=0.70\columnwidth]{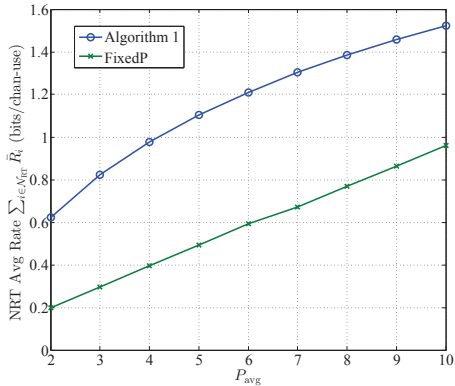}%
\caption{Sum of average throughput for all NRT users. The FixedP algorithm assigns a fixed power to all users set at $\Pmax$.}%
\label{Algorithm1_vs_FixedP_P2_10}%
\end{figure}

In Fig. \ref{Algorithm1_vs_FixedP_Qp1_8}, the sum of average NRT users' throughput is plotted while keeping $\Pavg=10$ but changing $q$. We can see that the FixedP algorithm results in a large degradation in the throughput compared to Algorithm \ref{Scheduling_Alg} which allocates the power and schedules the users optimally with respect to \eqref{Prob_DL}. The decrease in the throughput observed in both curves of Fig. \ref{Algorithm1_vs_FixedP_Qp1_8} is due to the increase in the parameter $q$. This increase makes constraint \ref{RT_QoS} more stringent and thus decreases the feasible region decreasing the throughput.

\begin{figure}%
\vspace{-0.2in}\centering
\includegraphics[width=0.7\columnwidth]{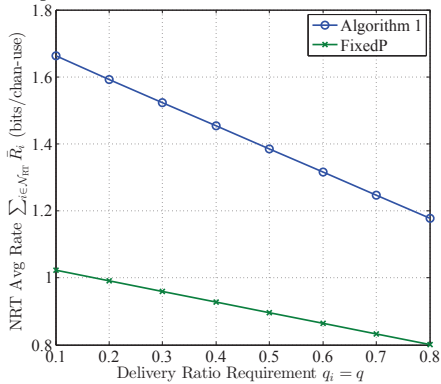}%
\caption{As $q$ increases, the RT users are assigned the channel more frequently. This comes at the expense of the NRT's throughput. However, the proposed algorithm outperforms the FixedP algorithm.}%
\vspace{-0.1in}\label{Algorithm1_vs_FixedP_Qp1_8}
\end{figure}


\section{Conclusions}
\label{Conclusion}
We discussed the problem of throughput maximization in downlink cellular systems in the presence of RT and NRT users. We formulated the problem as a joint power-allocation-and-scheduling problem. Using Lyapunov optimization theory, we presented an algorithm to optimally solve the throughput maximization problem. While we showed that the NRT power allocation is water-filling-like, the RT power allocation has a totally different structure that we provide in a closed-form expression and refer to as the ``Lambert Power Allocation''. The proposed algorithm is shown to have a polynomial-time complexity. Our simulations show that the proposed algorithm outperforms existing ones.




\bibliographystyle{IEEEbib}
\bibliography{MyLib}

\end{document}